%% file: main.tex
\documentclass[a4paper,10pt]{article}
\usepackage[utf8]{inputenc}

\usepackage[margin=1in]{geometry}

\usepackage{xfrac}

\usepackage{authblk}

\input{preamble/algorithms}
\input{preamble/maths}
\input{preamble/cref}

\title{Willy Wonka Mechanisms}

\date{February 13, 2024}

\author{
Thomas Archbold \quad
Bart de Keijzer \quad
Carmine Ventre}

\affil{King's College London\\United Kingdom\\
\texttt{
\string{%
\href{mailto:thomas.archbold@kcl.ac.uk}{thomas.archbold},%
\href{mailto:bart.de_keijzer@kcl.ac.uk}{bart.de\_keijzer},%
\href{mailto:carmine.ventre@kcl.ac.uk}{carmine.ventre}%
\string}%
@kcl.ac.uk}}

\begin{document}

\maketitle

\begin{abstract}
Bounded rationality in mechanism design aims to ensure incentive-compatibility for agents who are cognitively limited.
These agents lack the contingent reasoning skills that traditional mechanism design assumes, and depending on how these cognitive limitations are modelled this alters the class of incentive-compatible mechanisms.
In this work we design mechanisms without any ``obvious'' manipulations for several auction settings that aim to either maximise revenue or minimise the compensation paid to the agents.
A mechanism without obvious manipulations is said to be \emph{not obviously manipulable (NOM)}, and assumes agents act truthfully as long as the maximum and minimum utilities from doing so are no worse than the maximum and minimum utilities from lying, with the extremes taken over all possible actions of the other agents.
We exploit the definition of NOM by introducing the concept of \emph{golden tickets} and \emph{wooden spoons}, which designate bid profiles ensuring the mechanism's incentive-compatibility for each agent.
We then characterise these ``Willy Wonka'' mechanisms, and by carefully choosing the golden tickets and wooden spoons we use this to design revenue-maximising auctions and frugal procurement auctions.
\end{abstract}

\section{Introduction}

\label{sec:introduction}

In algorithmic mechanism design we are tasked with designing systems that elicit some information from a set of agents in order to return an outcome from some feasible set.
Each agent holds some piece of private information reflecting the true state of the world, and agents are assumed to be selfish, meaning that they may lie about their piece of information (also known as their type) if it is beneficial to do so.
The mechanisms we design must therefore be robust against the selfish actions of the agents and carefully consider the agents' incentives.
These incentives are modelled by solution concepts.
Strategyproofness, otherwise known as dominant-strategy incentive-compatibility, is a solution concept which stipulates that for each agent it is a dominant strategy to report her true type to the mechanism over any deviation.
This is desirable from a theoretical standpoint since if a mechanism is strategyproof then it guarantees that regardless of the joint action profile of the other players an agent will always maximise her utility by reporting her type truthfully rather than by lying to the mechanism.
Strategyproofness is a viable solution concept only when agents are \emph{perfectly} rational: in order to conclude that acting truthfully is indeed optimal an agent must have a detailed knowledge of the mechanism and correctly reason about all possible states that may result from her actions by considering the actions of the other agents.
This may place too high a cognitive burden on the agents to be useful in practical mechanisms.

Bounded rationality has seen growing interest in the mechanism design literature and considers agents who are cognitively limited.
\cite{Li2017} introduced a strengthening of strategyproofness known as ``obvious strategyproofness'' which assumes that agents will only tell the truth if it is obvious to do so.
Roughly a mechanism is \emph{obviously strategyproof (OSP)} if the minimum utility that an agent can achieve from telling the truth is no worse than the maximum utility that she can achieve from reporting her type dishonestly.
Troyan and Morrill take a more optimistic approach towards cognitively-limited agents, motivated by evidence \cite{Pathak2008,Dur2018} gathered from mechanisms in practice for school-choice and two-sided matching that agents in manipulable (i.e., non-strategyproof) mechanisms might fail to recognise when it is beneficial to lie, and introduce a weaker version of incentive compatibility for these agents known as ``non-obvious manipulability'' \cite{Troyan2020}.
Under this solution concept it is assumed that agents will only lie if it is obvious, meaning that the maximum (respectively, minimum) utility resulting from a lie is strictly greater than the maximum (respectively, minimum) utility resulting from reporting truthfully to the mechanism; otherwise, agents are assumed to report their types truthfully.
The extrema of the agents' utility functions are taken over all possible joint profiles of the other agents.
When a mechanism has no such ``obvious manipulations'' then it is said to be \emph{not obviously manipulable (NOM)}.

This definition is motivated by empirical evidence \cite{Pathak2008,Dur2018} from school choice and two-sided matching that people facing manipulable (i.e., non-strategyproof) mechanisms fail to recognise when to lie.
Students participating in the Boston Mechanism, for example, can guarantee a spot at their second-choice school by misreporting her preferences, whereas hospitals may end up being allocated a doctor it finds undesirable when attempting to manipulate the Deferred Acceptance algorithm.
With NOM thus defined they apply it to a variety of mechanism design settings both with and without monetary transfers, including school choice, two-sided matching, auctions, and bilateral trade.

In revenue maximisation, we are interested in designing incentive-compatible mechanisms that also provide guarantees on the profits that can be extracted from the bidders.
To analyse the performance of a mechanism in a prior-free setting we can compare the revenue that the mechanism extracts from the bidders against a variety of benchmarks.
For example, the ``optimal omniscient benchmark'' is a measure of how much profit the optimal mechanism extracts from the bidders when it knows the bids in advance, while the ``optimal fixed-price benchmark'' measures the revenue that can be extracted when each bidder is offered the same ``take-it-or-leave-it'' price.
Other benchmarks can be defined to reflect, for example, the revenue extracted by an auction that sets a single price and must sell at least some number of copies $m$ of the item, or the revenue extracted by an auction that must offer a \emph{monotone} price vector with respect to the ordering of the bidders according to their valuation.
\cite{Goldberg2006} introduced the notion of ``competitive'' auctions, adapted from the analysis of online algorithms, in the setting of prior-free auctions for digital goods.
An auction is said to be competitive to some benchmark if on every bid profile the revenue extracted by the mechanism is within a constant factor of that specified by the benchmark.
They focus on dominant strategy incentive compatibility and show that no truthful auction in this setting is competitive to even the optimal fixed-price benchmark, much less the optimal omniscient benchmark.
To this end, in this paper we study the extent to which we can design competitive auctions for digital goods and a variety of other settings by relaxing the strategyproofness requirement to non-obvious manipulability.

\medbreak 
\textbf{Related Work.}
\cite{Troyan2020} introduce non-obvious manipulability for direct-revelation mechanisms and provide a characterisation, similarly to \cite{Li2017}, for such mechanisms as those whose profitable deviations may be recognised by a cognitively limited agent that is unable to engage in contingent reasoning.
They apply this to the settings of school choice, two-sided matching, auctions, and bilateral trade and use their framework to classify mechanisms in these settings as either obviously manipulable or not obviously manipulable.

NOM has since been studied in a range of different contexts.
\cite{Aziz2021} look at obvious manipulations in the context of voting and study the conditions necessary for certain voting rules to be NOM.
They also look at computational issues related to computing obvious manipulations when they exist, reducing the problem to that of determining whether there exists a way for a set of voters to vote, given a set of votes that have already been submitted, in order to elect a given candidate, yielding a polynmomial-time algorithm for the $k$-approval voting rule.
\cite{Arribillaga2024} also study the non-obvious manipulability of voting rules and restrict attention to ``tops-only'' rules, which only consider each agent's top preference when selecting an outcome.
They first characterise the rules without obvious manipulations in this general setting, then restrict attention to two subclasses of these rules known as median voter schemes and voting-by-committees to provide a more fine-grained classification of the set of tops-only rule which are NOM.

\cite{Ortega2022} apply non-obvious manipulability to indirect mechanisms, specifically cake-cutting, and show that, unlike strategyproofness, NOM is compatible with a notion of fairness known as proportionality.
\cite{Psomas2022} also apply NOM to the setting of fair division and study deterministic mechanisms for allocating indivisible goods to agents with additive valuations.
They focus on envy-freeness up to one good (EF1) and show that while NOM mechanisms exist for maximising social welfare, the same is not true for maximising egalitarian or Nash welfare.
They also reduce the problem of designing NOM and EF1 mechanisms to that of designing EF1 algorithms, where the reductions preserve the efficiency guarantees.

\cite{AdKV2023a} study how to design NOM mechanisms using monetary transfers, providing characterisations for the class of allocation functions that are implementable in general domains, and then by focusing on single-parameter domains they recover an analogue to the monotonicity condition for allocation functions that are truthfully implementable.
They apply this to bilateral trade, and while any efficient, individually rational, budget balanced mechanism was known to be obviously manipulable, this issue persists even for approximate budget balance.
This line of work is extended in \cite{AdKV2023b} to study the class of allocation functions implementable with payments as indirect NOM mechanisms.
They prove an analogous result to the revelation principle for single-parameter agents: for any allocation function implementable by an indirect mechanism there is an equivalent direct mechanism that implements it.

\medbreak 
\textbf{Our Contribution.}
In this paper we study the extent to which it is possible to design NOM mechanisms that perform well for revenue maximisation and frugal procurement.
In the former setting we wish to allocate goods to the agents with the aim of maximising the social welfare as well as the revenue extracted from the bidders.
In the latter we instead allocate chores and must compensate the bidders for being selected, and we wish to design an incentive-compatible mechanism while minimising the sum of payments made to the bidders.

Starting with goods auctions in which agents are single-parameter we identify two necessary and sufficient conditions for mechanisms to be NOM in this context.
These conditions state that for each agent and each valuation there must be some valuation profile of the other bidders in which she wins the auction for free and in which she loses.
We refer to these bid profiles and ``golden ticket'' and ``wooden spoon'' profiles.
For single-parameter settings this characterisation holds for any allocation and payment functions with non-zero approximation guarantee to the optimal social welfare.

\begin{mainresult}
    An auction that $\alpha$-approximates social welfare and is $\beta$-competitive against social welfare for any $\alpha,\beta > 0$ is not obviously manipulable if and only if it is a Willy Wonka mechanism.
\end{mainresult}

We can show a slightly weaker result when generalising from binary allocation settings to more general outcome spaces.
Now the golden ticket profile is any input to the mechanism where an agent receives her most valued allocation for free, and this is both necessary and sufficient whenever the allocation function is maximal-in-range.
We show that the wooden spoon profile, again an input resulting in the agent losing the auction, is also sufficient to satisfy the worst-case NOM constraints.

We refer to mechanisms satisfying these golden ticket and wooden spoon properties as ``Willy Wonka'' mechanisms.
Given black box access to an $\alpha$-approximate (or maximal-in-range) allocation function we then design basic payment rules that result in high revenue for the auctioneer, which sacrifices the revenue from the lowest winning bidder in the approximately optimal allocation in order to satisfy the golden ticket property.
Specifically, given an input bid profile we first compute the approximately optimal allocation and charge each winning bidder first-price payments, except the lowest winning bidder whose payment is set to zero.
Of course, if we were to do this on each allocation then we might sacrifice the revenue from the only winning bidder (depending on the feasible set of allocations) and end up with zero revenue.
It suffices to have only one bid profile where this occurs in order to satisfy the best-case NOM constraints.
Therefore we can choose the bid profile corresponding to an allocation that maximises the number of winners as the golden ticket and thus minimise the amount of revenue we sacrifice.
For all remaining bid profiles we will simply recover the optimal revenue.

\begin{mainresult}
    There is a mechanism that is not obviously manipulable and $\alpha (1-1/\tau)$-competitive to the optimal revenue, where $\alpha$ denotes the approximation guarantee of the allocation function and $\tau$ the maximum number of winning bidders in an optimal allocation.
\end{mainresult}

Finally we turn our attention to procurement auctions where agents incur costs for being included in an allocation and agents must therefore be compensated.
Now we aim to minimise the social cost, and we want to design the payments to implement the socially optimal allocation as a NOM mechanism.
We want our payment rules to be frugal, meaning they should not overpay the agents to achieve incentive-compatibility.
To this end we apply our golden ticket and wooden spoon properties to this setting.
The former states that for each agent, on top of being compensated her bid (due to individual rationality), there must be a profile where she is paid an additional and costly sum of money that is sure to maximise her utility.
The latter, again, states that each must also lose the auction.
Our first mechanism is a natural analogue to the Willy Wonka mechanisms given for revenue maximisation, but we show that its performance with respect to frugality against the optimal solution can go to zero.
We then amend the payment rule to modify the (costly) golden ticket profiles and show that the resulting mechanism never overpays the agents by more than a factor of two with respect to the second-lowest social cost.

\begin{mainresult}
    There is a procurement auction that is not obviously manipulable and pays at most twice the cost of the second-best solution.
\end{mainresult}

We note that we change the benchmark from the optimum social cost to that of the second best social cost.
For a given agent, if they can never be allocated at the same time as another agent, then the frugality ratio with respect to the optimal will go to zero if the auction is NOM, and is therefore unavoidable.
We can sidestep this issue by benchmarking the mechanism against the cost of the next best solution, and show that now whenever an agent can only be the sole winner of the allocation our mechanism does not overpay at all with respect to this cost.
We conclude the paper with a brief discussion on the areas for further research that our work leaves open.

\section{Preliminaries}

\label{sec:preliminaries}

\textbf{General setup.}
For a natural number $k$, let $[k]$ denote the set $\{ 1, 2, \ldots, k \}$.
We consider the familiar mechanism design setting where there is a set $[n]$ of agents (or equally, bidders) where each agent $i$ has a piece of private information $t_i$ (also known as her type) taken from some domain $D_i$.
The set of type profiles is denoted by $D = \times_{i \in [n]} D_i$.
We also refer to $D$ as the set of \emph{bid profiles}.
For any $\bs{b} = (b_{1}, b_{2}, \ldots, b_{n}) \in D$ we write $(t_i,\bs{b}_{-i})$ to denote the bid profile obtained by replacing the $i$th coordinate of $\bs{b}$ with $t_i \in D_i$.

\medbreak

\textbf{Auctions.}
We consider mechanisms that use monetary transfers. 
In this setting an auction (mechanism) is a tuple $M = (a,p)$ consisting of an allocation function $a : D \to \mc{A}$, where $\mc{A}$ is some set of feasible allocations, and a payment function $p : D \to \R^{n}$.
On input $\bs{b}$ the output of $M$ is denoted $M(\bs{b})$.
An agent's type describes the utility she gets from each allocation, and we can think of agent $i$'s type as a function $t_{i} : \mc{A} \to \R$.
For a given allocation $A \in \mc{A}$ we assume $A$ can be decomposed into the vector $(A_1, A_2, \ldots, A_n)$, where $A_i$ denotes what is allocated to agent $i$ under $A$.
We also assume each agent derives utility only with regard to her own allocation, hence for each agent $i$ and allocation $A$ we have $t_{i}(A) = t_{i}(A_{i})$.
We will therefore denote by $\mc{A}_{i}$ the set of all unique allocations to agent $i$.
Moreover we denote by $\mu(A)$ the set of bidders allocated under the allocation $A$.
We will also refer to this set as the set of ``winners''.
The notation $\1[\cdot]$ represents the indicator function, and we will use $\1[i \in \mu(A)] = 1$ if $i \in \mu(A)$ and $\1[i \in \mu(A)] = 0$ otherwise.
Given a mechanism $M$ and agent $i$ with type $t_i$ we write $i$'s utility on outcome $M(\bs{b})$ as $\ui{i}{M(\bs{b})}{t_i}$.
All auctions we consider are assumed to be \emph{individually rational}, meaning that for each agent $i$ with type $t_i$ and every partial profile $\bs{b}_{-i}$ it holds that $\ui{i}{M(t_i,\bs{b}_{-i})}{t_i} \ge 0$.

\medbreak
\textbf{Goods and chores.}
Depending on the setting, we assume each agent receives either a non-negative or non-positive utility for each feasible allocation in $\mc{A}$.
We refer to the former as a goods auction and the latter as a procurement auction, or equivalently, an auction in which allocations represent (bundles of) chores.
In a goods auction we extract payments from the bidders and we assume our auctions make \emph{no positive transfers}, meaning no bidder is paid money in addition to their allocation.
Conversely, for a procurement auction agents incur a cost for being allocated and hence must be compensated, and here individual rationality implies that all payments to agents are non-negative.
We will formally introduce these notions, including those related to the optimality of an auction's allocation and payment functions, in their respective Sections \ref{sec:revenue} and \ref{sec:frugality} as needed.

\medbreak
\textbf{Incentive compatibility.}
We are interested in designing in\-cen\-tive-compatible mechanisms for agents with a particular form of imperfect rationality whereby they are only able to compare outcomes of the mechanism at the extremes of their utility function.
A mechanism $M = (a,p)$ is \emph{not obviously manipulable (NOM)} if the following two conditions hold for every $i \in [n]$:
\begin{gather}
  \label{eq:bnom}
  \sup_{\bs{b}_{-i}} \set{\ui{i}{M(t_i,\bs{b}_{-i})}{t_i}} \ge \sup_{\bs{b}_{-i}} \set{\ui{i}{M(b_i,\bs{b}_{-i}}{t_i}}, \\
  \label{eq:wnom}
  \inf_{\bs{b}_{-i}} \set{\ui{i}{M(t_i,\bs{b}_{-i})}{t_i}} \ge \inf_{\bs{b}_{-i}} \set{\ui{i}{M(b_i,\bs{b}_{-i}}{t_i}}.
\end{gather}
If \eqref{eq:bnom} holds then $M$ is \emph{best-case not obviously manipulable (BNOM)} and if \eqref{eq:wnom} holds then it is \emph{worst-case not obviously manipulable (WNOM)}.
We also use the term \emph{incentive compatible} to refer to a mechanism that is NOM.
Any misreport $b_{i}$ such that either \eqref{eq:bnom} or \eqref{eq:wnom} does not hold is said to be an \emph{obvious manipulation} of $M$.

\medbreak
\textbf{The structure of the feasible set.}
The structure of $\mc{A}$ reflects the auction setting under consideration.
We provide several examples here for clarity of exposition.
A set system is a tuple $(E,\mc{S})$ where $E$ is referred to as the ground set and $\mc{S} \subseteq 2^E$ the family of subsets.
In general in a \emph{binary allocation unit-demand} setting each bidder can either be allocated or unallocated in each allocation $A$, and the feasible set of allocations $\mc{A}$ in these settings can be described by the set $\mc{S}$ in the set system $([n],\mc{S})$.
This set $\mc{S}$ will vary depending on the type of auction being run, and can encompass, for example, digital goods auctions \cite{Goldberg2006}, knapsack auctions \cite{AggarwalHartline2006}, and spanning tree auctions \cite{Quadir2016}.
For these settings we can represent each allocation by the vector $A \in \{0,1\}^n$, where $A_i = 1$ if $i$ is allocated and $A_i = 0$ otherwise.
In a multi-unit auction there are $k$ copies of a single item and each agent may therefore be allocated up to $k$ copies (and note that agent types may not be linear functions in the number of items they are allocated).
In this case each allocation can be represented by a vector $A \in \{0,\ldots,k\}^n$, where $A_i$ describes the number of copies allocated to agent $i$ (and clearly we require $\sum_{i \in [n]} A_i \le k$ for $A$ to be a feasible allocation).
As a final example, consider a \emph{combinatorial auction} in which there is a set $S$ of items and each allocation distributes a (possibly empty) bundle of items to each agent.
Here for each allocation $A$ the entry $A_i \subseteq S$ describes the bundle of items allocated to $i$, and therefore the feasible set of allocations can be represented by the set of all $A$ such that $A_i \subseteq S$ for each $i$, and $A_i \cap A_j = \empty$ for each distinct $i,j \in [n]$.
Regardless of the specific setting we assume that for each player there is some feasible allocation in $\mc{A}$ where they do not win, otherwise clearly there is nothing to be done from the perspective of realigning incentives.

\medbreak
\textbf{Bidding languages.}
In general we assume that each agent $i$ reports her type directly to the mechanism by specifying a real number for each unique personal allocation $A_i \in \mc{A}_i$.
To avoid cluttering the notation we will drop the $i$ subscript from $A_i$.
Therefore agent $i$'s type is represented by the vector $t_i = (t_i^A)_{A \in \mc{A}_i}$.
In \Cref{sec:revenue} we assume that for each $A \in \mc{A}$ we have $t_{i}^A \in [0,h]$ and in \Cref{sec:frugality} we assume that $t_{i}^A \in [-h,0]$ for some $h \in \R_+$.
In auctions based on set systems each agent needs only to report a single number to the mechanism, and as such we refer to them in this case as being \emph{single-parameter}.
When this is not the case then agents are in general said to be \emph{multiparameter}.

\section{Revenue Maximisation}
\label{sec:revenue}

In this section we consider goods auctions, where for each agent $i$ and each allocation $A \in \mc{A}$ we have $t_i(A) \ge 0$.
We therefore refer to each agent's type as her \emph{valuation function} and denote the valuation function for agent $i$ as $v_i : \mc{A} \to \R_{\ge 0}$.
This is simply a notational change to differentiate between the two settings of revenue maximisation and frugality.
In this setting the auction mechanism extracts payments from bidders in exchange for their allocation, hence for an auction mechanism $M = (a,p)$ and input bid profile $\bs{b}$ the utility that agent $i$ with valuation $v_i$ will receive is denoted $\ui{i}{M(\bs{b})}{v_i} = v_i(a_i(\bs{b})) - p_i(\bs{b})$.

In this setting we consider allocation functions that (approximately) maximise the social welfare.
Fix a valuation profile $\bs{v}$.
The social welfare of an allocation $A$ with respect to $\bs{v}$ is simply $\SW{A,\bs{v}} = \sum_{i \in [n]} v_i(A_i)$.
An allocation function $a$ is said to $\alpha$-approximate the social welfare for $\alpha \in (0,1)$ if $\SW{a(\bs{v}),\bs{v}} \ge \alpha W^{*}$ for every $\bs{v}$, where $W^* = \max_{A \in \mc{A}} \SW{A,\bs{v}}$ denotes the optimal social welfare of an allocation with respect to $\bs{v}$.

We study the revenue we can extract from the bidders while ensuring the mechanism is NOM.
The revenue of an auction $(a,p)$ on bid profile $\bs{b}$ is simply the sum of the payments $\sum_{i \in [n]} p_i(\bs{b})$.
The performance of a mechanism's payment function will be measured with respect to some benchmark.
Let $\mc{X}(\bs{b})$ be the revenue of some payment benchmark $\mc{X}$ when given bid profile $\bs{b}$.
A payment function $p$ is said to be $\beta$-competitive against $\mc{X}$ for $\beta \in (0,1)$ if for all $\bs{b}$ it holds that $\sum_{i \in [n]} p_i(\bs{b}) \ge \beta \mc{X}(\bs{b})$.

\subsection{Binary allocation}

We begin by characterising the class of auctions which are NOM for single-parameter agents binary allocation settings.
Our characterisation shows that two properties based on the outcomes of the mechanism, known as ``golden ticket'' and ``wooden spoon'' profiles, are necessary and sufficient in order for the mechanism to be NOM.
The former states that, for each bidder and each bid, it is possible, given some bid profile of the other agents, to win the auction for free.
It is straightforward to see that this ensures that \eqref{eq:bnom} is satisfied, since there is a bid profile where each agent receives her highest possible utility (assuming no positive transfers from the mechanism to the agent).
On the other hand, as long as the allocation function has a positive welfare guarantee then the golden ticket property is necessary when the mechanism is BNOM, otherwise she may submit a bid below the assumed minimum payment price and get strictly greater utility in the best case than when bidding truthfully.

\begin{lemma}[Golden ticket for binary allocation]
    \label{lem:golden-ticket-binary-allocation}
    Let $M = (a,p)$ be an auction where the allocation function has any positive approximation $\alpha > 0$ to the optimal social welfare.
    Then $M$ is BNOM if and only if for every agent $i \in [n]$ with valuation $v_i$ there exists a bid profile such that $i$ wins the auction for free.
\end{lemma}

\begin{proof}
    ($\implies$)
    Begin by observing that $i$ must always be able to win the auction.
    On the bid profile $(v_i, 0, \ldots, 0)$ a positive approximation to the optimal social welfare is only possible if $i$ is allocated.
    Suppose for contradiction that whenever $i$ wins the auction she pays at least some positive price $p$.
    Her best utility when bidding $v_i$ truthfully is at most $v_i - p$.
    Now let $b_i$ be any misreport such that $0 < b_i < p$ and note that $p \le v_i$ by individual rationality.
    Take the profile $(b_i, 0, \ldots, 0)$ and by the same argument $i$ must win.
    Her utility for this misreport is at least $v_i - b_i \ge v_i - p$, breaking BNOM.
    Therefore $i$ must be able to win the auction for free.
    
    ($\impliedby$)
    Since $i$ can win the auction for free when bidding $v_i$ her best truthful utility is exactly $v_i$.
    This is clearly unbeatable by any misreport if $M$ makes no positive transfers, hence the auction is BNOM.
\end{proof}

Note that as a corollary of \Cref{lem:golden-ticket-binary-allocation} we have that for NOM auctions in which for each allocation the set of winning bidders is a singleton, such as single-item auctions, then the revenue guarantee against the optimal goes to zero when taken over all possible inputs.
To verify that this is the case, consider the Vickrey auction, which is strategyproof and hence also NOM.
The revenue obtained by this mechanism is simply the second highest bid, which is zero whenever there is a single agent $i$ with a positive bid $b_i$ and the remaining agents bid zero.
The optimal revenue is $b_i$, resulting in a $0$-approximation.

Analogously to the golden ticket, the wooden spoon property states that for each bidder and each bid, it is possible to lose the auction.
This guarantees that \eqref{eq:wnom} is satisfied, since now there is a bid profile where each agent receives zero utility and by individual rationality this is the lowest possible when bidding truthfully.
Since this holds for each bid then the worst utility of a dishonest bid can certainly be no greater than zero.
Now to prove the necessity of the wooden spoon we assume the payment function extracts any positive factor of revenue from the bidders and show that if the agent always wins the auction then it is possible to underbid and get a strictly greater worst case utility versus a truthful bid.

\begin{lemma}[Wooden spoon for binary allocation]
    \label{lem:wooden-spoon-binary-allocation}
    Let $M = (a,p)$ be an auction where the payment function guarantees to extract any positive factor $\beta > 0$ of the sum of bids in revenue from the bidders.
    Then $M$ is WNOM if and only if for every agent $i \in [n]$ with valuation $v_i$ there exists a bid profile such that $i$ loses the auction.
\end{lemma}

\begin{proof}
    ($\implies$)
    Suppose bidder $i$ always wins the auction.
    When bidding $b_i < v_i$ her worst utility is at least $v_i - b_i$.
    On the profile $(v_i, 0, \ldots, 0)$ then $i$ must win the auction and pay at least $\beta v_i$ meaning her worst truthful utility is at most $(1-\beta) v_i$.
    Taking a misreport $b_i < \beta v_i$ results in a worst case utility of $v_i - b_i > (1-\beta) v_i$, breaking WNOM.
    Therefore $i$ must lose the auction.

    ($\impliedby$)
    Take bidder $i$ with value $v_i$.
    Since $i$ can lose then her worst utility when bidding $v_i$ truthfully is at exactly 0.
    If $i$ were to underbid then the same argument gives her a worst dishonest utility of 0, while submitting an overbid $b_i$ can clearly yield a negative utility if $i$ were to win and pay $b_i > v_i$.
    Hence $M$ is WNOM.
\end{proof}

With \Cref{lem:golden-ticket-binary-allocation,lem:wooden-spoon-binary-allocation} we have a full characterisation for NOM auctions in a range of binary allocation settings.
Namely, as long the allocation function has some positive approximation guarantee to the optimal social welfare and the payment function always extracts some positive revenue from the bidders, then for each agent there must be both a bid profile where she wins the auction for free and a bid profile where she loses.
This motivates the following definition for auctions that satisfy these criteria.

\begin{definition}[Willy Wonka mechanism for binary allocation]
    \label{def:willy-wonka-binary}
    A \emph{Willy Wonka mechanism for binary allocation settings} is a mechanism $M = (a,p)$ whose allocation and payment functions satisfy the following:
    \begin{enumerate}
        \item For each $i$ and each $b_i$ there is a bid profile $\bs{b} = (b_i,\bs{b}_{-i})$ where $i$ wins the auction for free.
        \item For each $i$ and each $b_i$ there is a bid profile $\bs{b} = (b_i,\bs{b}_{-i})$ where $i$ loses the auction.
    \end{enumerate}
\end{definition}

Given a bid $b_i$ of player $i$ we refer to the \emph{golden ticket profile for $b_i$} as the profile $\bs{b}_{-i}$ such that $i$ wins for free on $\bs{b} = (b_i,\bs{b}_{-i})$.
Likewise we refer to the \emph{wooden spoon profile for $b_i$} as the profile $\bs{b}_{-i}$ such that $i$ loses the auction on $\bs{b} = (b_i,\bs{b}_{-i})$.
Therefore as long as each player $i$ has both a golden ticket profile and a wooden spoon profile for each bid $b_i$ then the resulting mechanism is NOM.
In the following we will assume that the wooden spoon profile is ensured by the feasible set, meaning that for a sufficiently approximately-optimal allocation function and each bid of an agent $i$ there is a bid profile of the other agents where $i$ must lose the auction.

Our idea is to design Willy Wonka mechanisms that sacrifice a small portion of the revenue in order to achieve incentive compatibility.
Ideally, for each bidder $i$ and each bid $b_i$ we want to find a bid profile $\bs{b}_{-i}$ and an allocation $A$ where $i$ wins that maximises the social welfare of the bidders excluding $i$, since we can charge first-price payments to the winners and award $i$ her allocation for free.
This $\bs{b}_{-i}$ would then be the golden ticket profile for $i$ with bid $b_i$.
In other words for each $i$ with bid $b_i$ the optimal Willy Wonka mechanism would compute an allocation that obtains revenue
\begin{equation}
    \label{eq:optimal-ww-revenue}
    \max_{A \, : \, i \in \mu(A)} \set{ \max_{\bs{b}_{-i}} \, \SW{A,\bs{b}_{-i}} }.
\end{equation}

Suppose we have black box access to an $\alpha$-approximate allocation function.
In order to tackle finding a good approximation to the optimal revenue above, consider the Willy Wonka mechanism that first computes an approximately optimal allocation and then charges first-price payments to all but the lowest winning bidder, who is allocated at no charge.
Then the revenue extracted by the mechanism is simply
\begin{equation*}
    \label{eq:next-best-willy-wonka}
    \alpha \cdot \SW{A,\bs{b}} - \min_{i \in \mu(A)} b_i.
\end{equation*} 

A problem with this mechanism is that it specifies a golden ticket profile for every input $\bs{b}$.
Now if there is any allocation $A$ such that $|\mu(A)|=1$ then we lose all of the revenue, resulting in a $0$-approximation. 
We amend this in \Cref{mech:ww-revenue-binary} by only allocating the lowest bidder for free when the cardinality of the winning set is maximal over all allocations in the image of the allocation function.

\begin{mechanism}
    \caption{Willy Wonka auction for binary allocation.}
    \label[mechanism]{mech:ww-revenue-binary}
    \begin{algorithmic}
        \REQUIRE Bid profile $\bs{b} \in D$.
        \ENSURE Allocation $a(\bs{b}) \in \mc{A} \subseteq \{0,1\}^n$, payments $p(\bs{b}) \in \R_{\ge 0}^n$.
        \STATE Let $A$ be an $\alpha$-optimal allocation w.r.t. $\bs{b}$ and let $\mu = \mu(A)$.
        \FOR{each bidder $i \in [n]$}
        \STATE $a_i(\bs{b}) \gets A_i$
        \STATE $p_i(\bs{b}) \gets b_i \cdot A_i$
        \ENDFOR
        \IF{$A \in \argmax \set{|A'|}{A' \in \text{image}(a)}$}
        \STATE Let $i \in \argmin \set{b_j}{j \in \mu}$
        \STATE $p_i(\bs{b}) \gets 0$
        \ENDIF
        \RETURN $a(\bs{b}), p(\bs{b})$
    \end{algorithmic}
\end{mechanism}

We note a few assumptions made by the mechanism in order to simplify its description.
Firstly, we assume that in each bid profile corresponding to a winning set of maximum cardinality there is a distinct lowest winning bidder in the optimal allocation.
If lowest winning bids are tied then in order to obtain the revenue bound of \Cref{thm:revenue-bound-binary} we would want for only one of these to be allocated for free.
We can break ties arbitrarily when this is the case, and since there is at least one profile where each player is the unique lowest winning bidder for each possible then \Cref{lem:golden-ticket-binary-allocation} ensures the mechanism is BNOM.

Again we note that this payment rule taken together with the allocation algorithm does not necessarily define any wooden spoon profiles.
In most cases the wooden spoon profiles should be ensured by the welfare approximation of the allocation function being high enough and by the feasible set of allocations.
In other words, if it is possible to lose the auction then WNOM is guaranteed by \Cref{lem:wooden-spoon-binary-allocation}.
In settings where this does not hold and it is always possible to allocate a given agent, then in order to achieve WNOM we must specify these profiles explicitly.
In this special case then, again, in order to achieve the revenue bound of \Cref{thm:revenue-bound-binary} we have to take care that the golden ticket profile of one agent does not coincide with the wooden spoon profile of another agent.

\begin{theorem}
    \label{thm:revenue-bound-binary}
    \Cref{mech:ww-revenue-binary} is not obviously manipulable and is $\alpha (1-1/\tau)$-competitive to the optimal revenue, where $\tau = \max \set{|\mu(A)|}{A \in \image(a)}$ denotes the cardinality of the largest winning set under $a$.
\end{theorem}

\begin{proof}
    The mechanism is NOM by \Cref{lem:golden-ticket-binary-allocation,lem:wooden-spoon-binary-allocation}.
    It charges first-price payments to all winning bidders apart from the lowest when the number of winners is maximal, meaning the revenue is a $(1-1/\tau)$ factor of the social welfare of the returned allocation.
    Since the allocation function is $\alpha$-approximate then the social welfare of the allocation is at least an $\alpha$ fraction of the optimal, giving the bound.
\end{proof}

In the next section we generalise the arguments in \Cref{lem:golden-ticket-binary-allocation,lem:wooden-spoon-binary-allocation} to handle general outcome spaces, with an additional condition on the allocation function.

\subsection{General outcome spaces}

In this section we assume $\mc{A}$ represents more general classes of allocations and that agents are multiparameter.
Recall that each agent $i$'s type is now a vector $(v_i^A)_{A \in \mc{A}}$ where $v_i^A \in [0,h]$ represents the utility $i$ receives from allocation $A$.
In order to derive similar golden ticket and wooden spoon properties to the previous section we impose an additional requirement on the allocation function of the mechanism that requires the allocation to be able maximise over a fixed range of allocations.
An allocation function $a$ is said to be \emph{maximal-in-range} if there exists some subset $\mc{R} \subseteq \mc{A}$ such that for all $\bs{b} \in D$ the function outputs $a(\bs{b}) \in \argmax_{A \in \mc{R}} \SW{A,\bs{b}}$.

\begin{lemma}[Golden ticket for general outcome spaces]
    \label{lem:golden-ticket-general}
    Let $M = (a,p)$ be an maximal-in-range auction with range $\mc{R} \subseteq \mc{A}$. 
    Then $M$ is BNOM if and only if for every agent $i \in [n]$ with valuation $v_{i} = (v_{i}^{A})_{A \in \mc{A}}$ there exists a bid profile $\bs{b}$ such that $a(\bs{b}) \in \argmax_{A \in \mc{R}} v_{i}^{A}$ and $p_{i}(\bs{b}) = 0$.
\end{lemma}

\begin{proof}
  Let $A \in \argmax_{A' \in \mc{R}} v_{i}^{A'}$ be $i$'s highest-valued allocation amongst those in $\mc{R}$ and first note that it must always be possible for $i$ to be allocated $A$.
  Take the bid profile $(v_{i}, \bs{0}, \ldots, \bs{0})$, where $\bs{0} = (0, \ldots, 0)$ denotes the all-zero type.
  Since $M$ is maximal-in-range and $A \in \mc{R}$ then $M$ must return the allocation $A$.

  ($\implies$)
  Now suppose that whenever the allocation $A$ is returned $i$ must pay some positive price $p$.
  Her best truthful utility is at most $v_{i}^{A} - p$.
  Let $b_{i}$ be any misreport of agent $i$ with $0 < b_{i}^{A} < p$ and $b_{i}^{A} = \max_{A' \in \mc{R}} b_{i}^{A'}$.
  Note that $b_{i}^{A}$ is an underbid for allocation $A$ since $p \le v_{i}^{A}$ by individual rationality.
  Now take the profile $(b_{i}, \bs{0}, \ldots, \bs{0})$.
  By the same argument as above $i$ must be allocated $A$ and pay some amount $p'$ where $0 < p' < b_{i}^{A}$.
  Her best utility for this misreport is at least $v_{i}^{A} - b_{i}^{A} > v_{i}^{A} - p$ and hence strictly larger than her best utility for reporting $v_{i}$ truthfully, breaking BNOM.
  Hence $i$ must win $A$ for free under some bid profile.

  ($\impliedby$)
  Since $i$ wins $A$ for free then bidding $v_{i}$ truthfully can result in a utility of $v_{i}^{A}$, which is clearly unbeatable since $M$ makes no positive transfers.
  Therefore $M$ is BNOM.
\end{proof}

We must strengthen the requirement of the allocation function being $\alpha$-optimal to being maximal-in-range.
This ensures that bidder $i$ is allocated her most valuable allocation in the range of the auction when the other bidders are excluded from winning.
Without this requirement then the auction could return any allocation $A$ that $i$ values at least an $\alpha$ fraction as much as her favourite allocation, which may not imply incentive-compatibility.
In this case the performance of the allocation function can be expressed as
\begin{equation}
    \label{eq:mir-alpha}
    \alpha = \min_{\bs{b}} \frac{\argmax_{A \in \mc{R}} \SW{A,\bs{b}}}{\argmax_{A \in \mc{A}} \SW{A,\bs{b}}}.
\end{equation}

The wooden spoon property remains exactly the same and using the exact very simple argument from the proof of \Cref{lem:wooden-spoon-binary-allocation} we can show that if each agent can lose the auction then it must be WNOM.

\begin{lemma}[Wooden spoon for general outcome spaces]
    \label{lem:wooden-spoon-general}
    Let $M = (a,p)$ be an auction for multiparameter agents that is $\beta$-competitive against the social welfare.
    If for every agent $i \in [n]$ with valuation $v_{i} = (v_{i}^{A})_{A \in \mc{A}}$ there exists a bid profile $\bs{b}$ such that $a_{i}(\bs{b}) = \emptyset$ then $M$ is WNOM.
\end{lemma}

We note that unlike \Cref{lem:golden-ticket-binary-allocation,lem:golden-ticket-general,lem:wooden-spoon-binary-allocation} this is not a complete characterisation since we do not have the necessity of the wooden spoon for auctions in these settings to be WNOM.
However, when designing our NOM auctions in the following we will only need the sufficiency of the golden tickets and wooden spoons to show the resulting mechanisms are incentive-compatible.

\begin{definition}[Willy Wonka mechanism for general outcome spaces]
    \label{def:ww-revenue-general}
    A \emph{Willy Wonka mechanism for general outcome spaces} is a mechanism $M = (a,p)$ whose allocation and payment functions satisfy the following:
    \begin{enumerate}
        \itemsep0em
        \item For each $i$ and each $b_i$ there is a bid profile $\bs{b} = (b_i,\bs{b}_{-i})$ where $i$ is allocated $\argmax_{A \in \mc{A}} b_i^A$ for free.
        \item For each $i$ and each $b_i$ there is a bid profile $\bs{b} = (b_i,\bs{b}_{-i})$ where $i$ loses the auction.
    \end{enumerate}
\end{definition}

We only need to modify \Cref{mech:ww-revenue-binary} slightly in order to apply to general outcome spaces, which we provide in \Cref{mech:ww-revenue-general}.
We now assume the allocation function $a$ is maximal-in-range as per \Cref{lem:golden-ticket-general}, but as before in \Cref{mech:ww-revenue-general} we compute an approximately optimal allocation and then charge first-price payments in most cases.
When the allocation has maximal cardinality then we sacrifice the revenue from the lowest winning bidder in order to satisfy the golden ticket property.

\begin{mechanism}[h]
    \caption{Willy Wonka auction for general outcome spaces.}
    \label[mechanism]{mech:ww-revenue-general}
    \begin{algorithmic}
        \REQUIRE Bid profile $\bs{b} \in D$.
        \ENSURE Allocation $a(\bs{b}) \in \mc{A}$, payments $p(\bs{b}) \in \R_{\ge 0}^n$.
        \STATE Let $A \in \argmax \set{\SW{A',\bs{b}}}{A' \in \mc{R}}$ and let $\mu = \mu(A)$.
        \FOR{each bidder $i \in [n]$}
        \STATE $a_i(\bs{b}) \gets A_i$
        \STATE $p_i(\bs{b}) \gets b_i^A \cdot \1[i \in \mu]$
        \ENDFOR
        \IF{$A \in \argmax \set{|A'|}{A' \in \mc{R}}$}
        \STATE Let $i \in \argmin \set{b_j^A}{j \in \mu}$
        \STATE $p_i(\bs{b}) \gets 0$
        \ENDIF
        \RETURN $a(\bs{b}), p(\bs{b})$
    \end{algorithmic}
\end{mechanism}

Again it is assumed that the wooden spoon profiles arise from the feasible set itself, that is, for each agent there is some bid profile such that it is not feasible to allocate the agent.
If this is not the case we would have to artificially exclude the agent from the winning set.
Naturally \Cref{mech:ww-revenue-general} obtains a similar revenue guarantee to \Cref{mech:ww-revenue-binary} of $\alpha (1-1/\tau)$ where $\tau$ is the size of the maximum cardinality winning set for any allocation returned by the mechanism, and $\alpha$ again corresponds to the approximation guarantee of the allocation function with respect to the socially optimal allocation and takes the form given in \eqref{eq:mir-alpha}.

\begin{theorem}
    \label{thm:revenue-bound-general}
    \Cref{mech:ww-revenue-general} is not obviously manipulable and is $\alpha (1-1/\tau)$-competitive to the optimal revenue, where $\tau = \max \set{|\mu(A)|}{A \in \mc{R}}$ denotes the cardinality of the largest winning set under the maximal-in-range allocation function $a$.
\end{theorem}

\section{Frugal Procurement}
\label{sec:frugality}

In this section we consider procurement auctions, where for each agent $i$ and each allocation $A \in \mc{A}$ we have $t_i(A) \le 0$.
For clarity of exposition we now introduce for each agent $i$ a \emph{cost function} and denote it $c_i : \mc{A} \to \R_{\ge 0}$, setting $c_i = -t_i$.
Now we have a procurement auction that allocates chores to the agents, hence the agents must be compensated for taking on their allocation.
For a procurement auction $M = (a,p)$ given the input $\bs{b}$ the utility that agent $i$ with cost $c_i$ will receive is denoted $\ui{i}{M(\bs{b})}{c_i} = p_i(\bs{b}) - c_i(a_i(\bs{b}))$.

We are now interested in allocation functions that (approximately) minimise the social cost of resulting allocation.
Fix a cost profile $\bs{c}$.
The social cost of an allocation $A$ with respect to $\bs{c}$ is $\SC{A} = \sum_{i \in [n]} c_i(A_i)$.
An allocation function $a$ is said to $\alpha$-approximate the social cost for $\alpha \ge 1$ if $\SC{a(\bs{c}),\bs{c}} \le \alpha C^*$ for every $\bs{c}$, where 
$C^* = \min_{A \in \mc{A}} \SC{A,\bs{c}}$ denotes the social cost of an optimal allocation with respect to $\bs{c}$.

We now need to compensate agents for their allocations.
All of our mechanisms must be individually rational, hence each bidder must always be paid at least the cost she reports for the resulting allocation.
We want to study the extent to which our mechanism overpays with respect to some benchmark.
Let $\mc{X}$ be some payment benchmark for a procurement auction and $\mc{X}(\bs{b})$ denote the sum of payments made to the bidders on bid profile $\bs{b}$.
A payment function $p$ is said to be $\beta$-frugal against $\mc{X}$ for $\beta \ge 1$ if $\sum_{i \in [n]} p_i(\bs{b}) \le \beta \mc{X}(\bs{b})$ for all $\bs{b}$.
We refer to $\beta$ as the \emph{frugality ratio} of the payment function (and equally, the mechanism) with respect to $\mc{X}$.

We think of the auctions as allocating bundles of chores, and to maintain individual rationality we must now instead reimburse or compensate the bidders allocated a chore with a monetary transfer.
For consistency we still refer to the set of such bidders as the set of winners, and every other bidder is therefore still considered a loser of the auction.
We also restrict our attention to single-parameter agents.

\begin{lemma}[Golden ticket for binary allocation]
    \label{lem:golden-ticket-frugality}
    Let $M = (a,p)$ be a procurement auction that $\alpha$-approximates the social cost for any $\alpha > 0$.
    Then $M$ is BNOM if and only if for every agent $i \in [n]$ with cost $c_i$ there exists a bid profile where $i$ wins the auction and is paid $h$.
\end{lemma}

\begin{proof}
    ($\implies$)
    Suppose for contradiction that $M$ only ever pays bidder $i$ at most $p < h$ when she bids $c_i$.
    By individual rationality we must have $p \ge c_i$.
    Now take an overbid $b_i$ such that $p < b_i < h / \alpha$.
    Let $\bs{b} = (b_i,\bs{b}_{-i})$ be a bid profile such that $b_j > \alpha b_i$ for each $j \neq i$.
    On $\bs{b}$ the chore must get allocated to $i$ in order to preserve the $\alpha$ guarantee to the optimal social cost.
    Since $i$ is allocated the chore she must be compensated at least $b_i$, resulting in a best case dishonest utility of at least $b_i - c_i$ which is strictly greater than her best case truthful utility of $p - c_i$, thus violating BNOM.
    
    ($\impliedby$)
    If there exists a bid profile where $i$ wins and is paid $h$ then when $i$ has cost $c_i$ and bids truthfully her best utility is $h - c_i \ge 0$.
    Since $i$ can never submit an overbid greater than $h$ then no misreport can get her a greater utility, thus $M$ is BNOM.
\end{proof}

\begin{lemma}[Wooden spoon for binary allocation]
    \label{lem:wooden-spoon-frugality}
    Let $M = (a,p)$ be a procurement auction that is $\beta$-frugal against the social cost for any $\beta > 0$.
    Then $M$ is WNOM if and only if for every agent $i \in [n]$ with cost $c_i$ there exists a bid profile such that $i$ loses the auction.
\end{lemma}

\begin{proof}
    ($\implies$)
    Assume $i$ always wins the auction.
    Then on an overbid $b_i > c_i$ she would get utility at least $b_i - c_i$.
    Now take the bid profile $(c_i, h, \ldots, h)$ such that $(n-1) h / \beta > c_i$ where only $i$ may win the auction.
    The $\beta$-frugality of the auction means that $i$ cannot be compensated more than $\beta c_i$ meaning her best truthful utility is at most $(\beta-1) c_i$.
    Taking $b_i > \beta c_i$ results in a worst utility of at least $b_i - c_i > (\beta-1) c_i$ thus breaking the WNOM constraints.
    Therefore $i$ must be able to lose the auction.
    
    ($\impliedby$)
    By a similar argument to the proof of \Cref{lem:golden-ticket-frugality} if the monetary payment is independent of the bid then $i$'s worst truthful and dishonest utilities are both $0$, satisfying WNOM. 
\end{proof}

We use the above conditions to design frugal Willy Wonka mechanisms for procurement auctions.
Similarly to \eqref{eq:optimal-ww-revenue} fix a player $i$ with bid $b_i$.
On the golden ticket profile for $i$ the optimal Willy Wonka mechanism for frugal procurement should select an allocation $A$ that compensates the agents by
\begin{equation*}
    \label{eq:optimal-ww-frugality}
    \min_{A \, : \, i \in \mu(A)} \set{ \min_{\bs{b}_{-i}} \, \SC{A,\bs{b}_{-i}} + h },
\end{equation*}
since $i$ will be paid $h$ while the remaining winners are paid their bid.
To obtain a good approximation to the above, again consider a mechanism with black box access to an $\alpha$-approximate allocation function such that on input $\bs{b}$ the mechanism selects an $\alpha$-optimal allocation $A$ that pays the highest winner $h$ and the remaining winners their bid.
Now the highest winning bidder (instead of the lowest) is the recipient of the golden ticket, since this minimises the total payment made on the given instance.
Then the total cost on allocation $A$ is equal to
\begin{equation*}
    \label{eq:second-best-frugality}
    \alpha \cdot \SC{A,\bs{b}} + h - \max_{i \in \mu(A)} b_i.
\end{equation*}

Again, we need not designate a golden ticket profile (giving the highest winning bidder a payment of $h$) for every bid profile $\bs{b}$ as this will degrade the frugality ratio with respect to the sum of (winning) bids.
Consider an analogous approach to \Cref{mech:ww-revenue-binary,mech:ww-revenue-general} for designing frugal Willy Wonka procurement auctions, where the golden ticket profile occurs for each bidder when they have the highest reported cost in an (approximately) optimal allocation of maximum winning set cardinality.
Suppose $\mc{A}$ contains all allocations with singleton winning sets.
Now fix bidder $i$ and take $\bs{b} = (b_i,h,\ldots,h)$ for some $b_i$.
Observe that the optimal cost of allocating only $i$ is $b_i$, and at least $h$ for any other allocation.
For a given $\alpha$ then a sufficiently small $b_i$ means $i$ must be the sole winner.
If the mechanism paid $h$ to bidder $i$ then the resulting frugality ratio would be infinite when we take $b_i$ to zero.

We provide an alternative way of specifying the golden ticket profiles of each bidder in \Cref{mech:ww-frugal-binary-2} and show that although its frugality ratio with respect to the optimal can still be infinite, it will never overpay the bidders by more than twice the cost of the second-best solution, provided there exists an allocation for each agent where they are not the sole winner.
Let $k = \max_{A \, : \, i \in \mu(A)} |\mu(A)|$ denote the maximum cardinality of winning set containing $i$ for any allocation in $\mc{A}$, and let $S$ be the corresponding winning set of any such allocation.
Given bid $b_i$ of player $i$ define the golden ticket profile such that for all $j \neq i$ we have $b_j = h/(k-1)$ whenever $j \in S$ and $b_j = h$ otherwise.
Denote this golden ticket for $b_i$ as $\gamma_i(b_i)$.

\begin{mechanism}[h]
    \caption{Willy Wonka procurement auction.}
    \label[mechanism]{mech:ww-frugal-binary-2}
    \begin{algorithmic}
        \REQUIRE Bid profile $\bs{b} \in D$.
        \ENSURE Allocation $a(\bs{b}) \in \mc{A} \subseteq \{0,1\}^n$, payments $p(\bs{b}) \in \R_{\ge 0}^n$.
        \STATE Let $A$ be an $\alpha$-optimal allocation w.r.t. $\bs{b}$ and let $\mu = \mu(A)$.
        \FOR{each bidder $i \in [n]$}
        \STATE $k \gets \max_{A \, : \, i \in \mu(A)} |\mu(A)|$ and $h' \gets h/(k-1)$ (when $k>1$)
        \medbreak
        \STATE $\gamma_i(b_i) \gets ( \underbrace{ h',\ldots,h'}_{(k-1)},\underbrace{h,\ldots,h}_{(n-k)} )$
        \medbreak
        \STATE $a_i(\bs{b}) \gets A_i$
        \STATE $p_i(\bs{b}) \gets b_i \cdot A_i$
        \medbreak
        \IF{$\bs{b}_{-i} = \gamma_i(b_i)$}
            \STATE $p_i(\bs{b}) \gets 0$
        \ENDIF
        \ENDFOR
        \RETURN $a(\bs{b}), p(\bs{b})$
    \end{algorithmic}
\end{mechanism}

We note that the frugality ratio for \Cref{mech:ww-frugal-binary-2} can vary depending on the size $k$ of the set $S$.
If $k > 1$ then the total payments made on the bid profile $(b_i,\gamma_i(b_i))$ is $h + (k-1) h / (k-1) = 2h$.
When $b_i = 0$ then the optimum cost is $h$, giving a frugality ratio of $2$.
If $k = 1$ then the sum of payments is $h$ while the optimum is zero, meaning the frugality ratio is infinite.

Instead let us benchmark the mechanism again the cost of the second-best solution.
We denote by $\text{FR}^{(2)}(\bs{b})$ the frugality ratio of a given mechanism on input $\bs{b}$ with respect to the cost of the second-best solution:
\begin{equation*}
    \label{eq:FR2}
    \text{FR}^{(2)}(\bs{b}) =
    \frac{\sum_{i \in [n]} p_i(\bs{b})}
         {\min_{A} \set{\SC{A,\bs{b}}}{\SC{A,\bs{b}} > C^*}},
\end{equation*}
where $C^* = \min_A \SC{A,\bs{b}}$ denotes the cost of an optimum allocation on $\bs{b}$.
We then take $\text{FR}^{(2)} = \max_{\bs{b}} \text{FR}^{(2)}(\bs{b})$.
With the following we show that \Cref{mech:ww-frugal-binary-2} has good performance with respect to $\text{FR}^{(2)}$ as long as there exists for each agent an allocation where they are not the only winning bidder.

\begin{claim}
    The frugality ratio $\text{FR}^{(2)}$ for \Cref{mech:ww-frugal-binary-2} is equal to $2$ whenever $k > 1$, and $1$ otherwise.
\end{claim}
\begin{proof}
    Let $A = a(\bs{b})$ be the allocation returned by $a$ on input $\bs{b} = (b_i,\gamma_i(b_i))$.
    For all cases where $k \ge 1$ then the cost of the second-best solution is at least $h$.
    When $k > 1$ then if $i \in \mu(A)$ the total sum of payments is at most $h + (k-1)h/(k-1) = 2h$, otherwise if $i \notin \mu(A)$ then the total sum of payments is simply $h$.
    When $k = 1$ then $\gamma_i(b_i) = (h,\ldots,h)$ by definition, and whoever wins is paid $h$, either because $i$ wins the golden ticket or because some bidder $j \neq i$ wins and must be compensated at least her bid.
\end{proof}

\section{Conclusion}
\label{sec:conclusion}

In this work we study NOM auctions for revenue-maximisation and frugal procurement.
We characterise the set of incentive-compatible mechanisms by introducing two design criteria that guarantee the resulting mechanism is NOM (in the right context), namely golden tickets and wooden spoons.
The former ensures a profile where each bidder can maximise her utility when reporting her type truthfully, and the latter ensures her worst utility when bidding truthfully is zero, and taken together they ensure the best-case and worst-case NOM constraints are always satisfied.
Given black-box access to an (approximately) optimal allocation function we then show that we can design simple payment rules that provide good performance with respect to either maximising the sum of payments extracted from the bidders or minmising the compensation paid.
These payment rules work by selecting, for an optimal allocation with the maximum (across all allocations) number of winning bidders, a single bidder whose utility function we will maximise.

There are several directions for further research left open by our work.
The performance of the payment rules we design is dependent on the number of winners in an optimal allocation.
Can we design payment rules whose performance does not depend on the structure of the feasible set, to better deal with instances where this results in poor performance for the mechanism?
We also note that, given a bid $b_i$ of player $i$, the best and worst outcomes of the mechanism always occur on the same profile, regardless of her true type.
These are the golden ticket and wooden spoon profiles, respectively, and the resulting mechanisms are therefore \emph{single-line} \cite{AdKV2023a}.
These payment schemes are simple to state.
Are golden ticket and wooden spoon payment rules the only types of single-line payment rules leading to good performance from the mechanism?
NOM also permits much more general payment schemes -- can we design non single-line payment rules that yield better performance that can be achieved by a single-line mechanism?
Finally, can we design (efficient) allocation functions ourselves for general settings so that our mechanisms do not rely on them as black boxes, and can we use in conjunction with an appropriate payment rule to yield a NOM mechanism?

\bibliographystyle{apalike}
\bibliography{references}

\end{document}

%% file: preamble/algorithms.tex
\usepackage{algorithm}
\usepackage[noend]{algorithmic}

\newenvironment{mechanism}[1][htb]{%
    \floatname{algorithm}{Mechanism}
    \begin{algorithm}[#1]%
}{\end{algorithm}}

%% file: preamble/maths.tex
\usepackage{amsmath}
\usepackage{amssymb}
\usepackage{amsthm}
\usepackage{dsfont}

\theoremstyle{definition}

\newtheorem{definition}{Definition}

\newtheorem{claim}{Claim}

\theoremstyle{plain}
\newtheorem{theorem}{Theorem}
\newtheorem{lemma}{Lemma}

\newcommand{\bs}{\boldsymbol}
\newcommand{\mc}{\mathcal}

\newcommand{\R}{\mathbb{R}}
\newcommand{\1}{\mathds{1}}

\newcommand{\ui}[3]{u_{#1}(#2 \, ; \, #3)}

\newcommand{\SW}[1]{\text{SW}(#1)}
\newcommand{\SC}[1]{\text{SC}(#1)}

\DeclareMathOperator*{\argmax}{argmax}
\DeclareMathOperator*{\argmin}{argmin}

\DeclareMathOperator*{\image}{image}

\usepackage{xparse}
\NewDocumentCommand{\set}{m g}{%
  \IfNoValueTF{#2}
    {%
        \{ \, #1 \, \}%
    }    
    {%
        \{ \, #1 \, : \, #2 \, \}%
    }%
}

\newcounter{mainresult}[section]
\newenvironment{mainresult}[1][]{%
    \refstepcounter{mainresult}\par\medskip
    \begin{center}
    \begin{minipage}[c]{0.85\linewidth}
       \noindent \textbf{Main Result~\themainresult #1 (informal).} \rmfamily}
    {\end{minipage} \end{center} \medskip}

%% file: preamble/cref.tex
\usepackage{hyperref}
\usepackage[nameinlink]{cleveref}

\crefname{mechanism}{mechanism}{mechanisms}
\crefname{claim}{claim}{claims}
\crefname{property}{property}{properties}